\documentclass{article}
\usepackage{a4}
\usepackage[centertags,sumlimits,intlimits,namelimits]{amsmath}
\usepackage{amsfonts,amsthm,amssymb}

\usepackage{chicago}

\newcommand{\Ical}{{\mathcal I}}

\renewenvironment{description}{\list{}{
\leftmargin 12pt \itemindent8pt
}
}{
  \endlist
}

\newcommand{\Levy}{L\'evy }
\newcommand{\Ito}{It\^o }

\newcommand{\half}{\frac{1}{2}}

\newtheorem{thm}{Theorem}[section]

\newtheorem{prop}[thm]{Proposition}

\newtheorem{col}[thm]{Corollary}

\theoremstyle{definition}
\newtheorem{defin}[thm]{Definition}

\theoremstyle{remark}

\newtheorem{rem}[thm]{Remark}

\newcommand{\Ind}{{ 1}}
\newcommand{\ind}[1]{\Ind_{\{#1\}}}

\newcommand{\E}{\mathbb{E}}
\newcommand{\RR}{\mathbb{R}}
\newcommand{\PP}{\mathbb{P}}
\renewcommand{\P}{\PP}

\newcommand{\cB}{{\mathcal B}}

\newcommand{\cF}{{\mathcal F}}

\newcommand{\cS}{{\mathcal S}}

\newcommand{\cI}{{\mathcal I}}

\newcommand{\stpr}[1]{(#1_t)_{t \ge 0}}

\renewcommand{\ind}[1]{\Ind_{\{#1\}}}

\newcommand{\squishlist}{
   \begin{list}{$\bullet$}
    { \setlength{\itemsep}{0pt}      \setlength{\parsep}{3pt}
      \setlength{\topsep}{3pt}       \setlength{\partopsep}{0pt}
      \setlength{\leftmargin}{1.5em} \setlength{\labelwidth}{1em}
      \setlength{\labelsep}{0.5em} }
      }

\newcommand{\Q}{\mathbb{Q}}

\newcommand{\R}{\RR}

\begin{document}

\title{{\sf CDO term structure modelling  with L\'evy processes and the relation to market models}\footnote{
 Support by the Polish Ministry of Science and Education projects: 5 1P03A 03429 ''Stochastic
 evoluation equations with L\'evy noise'' and N N201 419039  ''Stochastic equations in infinite dimensional spaces''is gratefully acknowledged. We thank the Fields Institute
 for the kind support as well as Frank Gehmlich and an anonymous referee for their helpful comments.}}
\author {Thorsten Schmidt\footnote{Technical University Chemnitz, Reichenhainer Str. 41, 09126 Chemnitz, Germany. Email:
thorsten.schmidt@mathematik.tu-chemnitz.de } and Jerzy Zabczyk
              \footnote{Institute of Mathematics of the Polish Academy of Sciences, Sniadeckich 8, P.O.B. 21, 00-956 Warszawa 10, Poland}}

\maketitle

\begin{abstract} This paper considers the modelling of collateralized debt obligations (CDOs).
We propose a top-down model  via forward rates generalizing Filipovi\'c, Overbeck and Schmidt (2009) to the case
where the forward rates are driven by a finite dimensional L\'evy process.
The contribution of this work is twofold:  we provide conditions for absence of arbitrage
in this generalized framework. Furthermore, we  study the relation  to market models by  embedding them in the forward rate framework
in spirit of \citeN{BGM}.

\textbf{Key words}:  collateralized debt obligations, loss process, single tranche CDO, term structure of forward spreads, Levy processes,  market models, Libor  rate.
\end{abstract}

\section{Introduction}
\label{sec:intro}

A \emph{collateralized debt obligation} (CDO) is a security backed by a pool of credit risky securities.
There are issued notes on so-called tranches of the CDO which are characterized by different levels
of credit riskiness or, in financial terms, seniorities.
In this work we consider a general model for the evolution of tranche prices similar in spirit to
the forward rate approach of \citeN{HJM} and derive conditions under which the model is free of arbitrage.

For practical applications, market models are
of a high importance. In this kind of models, traded securities have a finite set of maturities while in
the forward rate approach all possible maturities are considered. For the pricing of options respectivly calibration of the model one imposes a simple
dynamics for suitable quantities and uses the conditions for absence of
arbitrage which results in tractable pricing formulas.

The main goal of this work is to provide conditions
which guarantee the absence of arbitrage in a general class of models and to study market models
embedded in this framework.
The new contribution of this work is the following: first, we
consider models  where the driving random process is a
\Levy process, generalizing \citeN{FilipovicSchmidtOverbeck}. We
derive general drift conditions which ensure that the market is free
of arbitrage. In a next step we consider market models similar to those in \citeN{GrbacEberleinSchmidt}.
However, we embed the market models in the forward rate models and derive the resulting drift conditions.
The risk-free case studied in \citeN{BGM} turns out to be a special case.

\section{Collateralized Debt Obligations}
Consider a complete filtered probability space $(\Omega,\cF,(\cF_t)_{t \ge 0},\P)$ satisfying the usual conditions.
Mathematically speaking, a collateralized debt obligation is  a derivative on a portfolio of $N$ credit risky securities.
With each security there is an associated nominal and we assume that the total nominal is one.
Denote the process of accumulated losses over time by $L=\stpr{L}$.
Then $L$ is a pure-jump process which jumps upward at  defaults of the securities in the pool
and the jump size is  the occurring loss.
As the total nominal is one,  $L_t \in [0,1]$ for all $t \ge 0$.
A special case, quite often considered in practice, is when the loss on each default is a constant.
By $\cI\subset[0,1]$ we denote the set of attainable loss fractions and we assume that
$\cI=[0,1]$. The case where $\cI$ is finite may be considered analogously, see \citeN{FilipovicSchmidtOverbeck}.

Similar to \citeN{FilipovicSchmidtOverbeck} we consider \emph{$(T,x)$-bonds} as basic constituents:
 a $(T,x)$-bond pays $\ind{L_T \le x}$ at maturity $T$, for $x \in \cI$. Its price at time $t$ is denoted by $P(t,T,x)$.
For $x=1$ we obtain that $P(t,T,1)=:P(t,T)$ equals the risk-free bond. In \citeN{GrbacEberleinSchmidt}, Section 6.1, it is
shown how to extract $(T,x)$-bonds from observed CDO quotes under common assumptions.

\paragraph{Pricing CDOs with (T,x)-bonds.}
We consider $(T,x)$-bonds as they are
sufficient to provide prices for CDOs and similar derivatives as we will show in this section.
First of all,   $(T,x)$-bonds are sufficient for
pricing European derivatives on the loss process. Indeed, consider a
payoff $h(L_T)$ such that
$$ h(L_T)= h(1)-\int_0^1 h^\prime(y) \ind{L_T \le y} dy. $$
Linearity  of prices on contingent claims implies that the price at time $t$ of the  derivative offering $h(L_T)$ at $T$ is
given by the following functional of $(T,x)$-bonds
$$h(1) P(t,T) - \int_0^1 h^\prime(y) P(t,T,y) dy. $$
A more general result holds true. Investing in CDOs is
done via a so-called \emph{single-tranche CDO} (STCDO), sometimes also
called tranche credit default swap. A STCDO is represented by its
lower and upper detachment points, $x_1$ and $x_2$, with  $0\le x_1
< x_2 \le 1$. The investor receives coupon payments at  times $
T_1,\dots,T_n$. In exchange to this, the investor covers a certain
part of the occurring losses in each period. Set
$$ H(x):=(x_2-x)^+ - (x_1-x)^+ = \int _{(x_1,x_2]}\ind {x \le y} dy. $$
Then, investing in the STCDO with swap rate $S$ is equivalent to the following payment
stream:
\begin{enumerate}
\item \emph{Payment leg.} The investor receives  $S\, H(L_{T_i})$ at times $T_i, i=1,\dots,n$.
\item \emph{Default leg.} The investor pays $-dH(L_t)=H(L_{t-})-H(L_t)$ at any time where $\Delta L_t \not = 0$.
\end{enumerate}
In \citeN{FilipovicSchmidtOverbeck} it is shown that the value of the STCDO at time $t$ can be derived solely
on the basis of $(T,x)$-bonds. In the case where risk-free and risky $(T,x)$-bonds are independent, it follows from
Lemma 4.1 therein that the value of the STCDO at time $t$ is given by
 \begin{align*}
V(t,S)= \int_{(x_1,x_2]}\Bigg( &S \sum_{i=1}^n P(t,T_i,y)+ P(t,T_n,y)-P(t,T_0,y) \\
&+\int_{T_0}^{T_n}  f(t,u) P(t,u,y)   du   \Bigg)dy,
\end{align*}
where $f(t,u)$ denotes the risk-free forward rate.
Setting $V=0$ and solving for $S$ gives the par-spread  for this investment. Market prices
of STCDOs are typically quoted via the par-spread.


\section{Arbitrage-free term structure movements}
In this article we consider term structure movements of $(T,x)$-bonds given by
\begin{align} \label{eq:TxbondsViaForwardRates}
 P(t,T,x) = \ind{L_t \le x} \exp \Big(-\int_t^T f(t,u,x) du \Big),
 \end{align}
 where $f(t,T,x)$ is the $(T,x)$-\emph{forward rate} prevailing at $t$.
Let us assume that
\begin{description}
\item[(A1)] $L_t= \sum_{s\le t} \Delta L_s$ is c\`adl\`ag, non-decreasing, adapted, pure jump process,
which admits an absolutely continuous compensator $\nu^L(t,dx)dt$ satisfying
$\int_0^t\nu^L(s,\cI)ds<\infty$ (finite activity).
\end{description}
As shown in \citeN{FilipovicSchmidtOverbeck}, under {(A1)}, the indicator process $(\ind{L_t\le x})_{t \ge 0}$ is c\'adl\'ag and
has intensity
\begin{equation}\label{deflambda}
  \lambda(t,x)= \nu^L(t,(x-L_{t},1]).
\end{equation}
  That is,
  \begin{equation}\label{eqMx}
    M^x_t= 1_{\{ L_t\le x\}} +\int_0^t 1_{\{ L_{s}\le x\}}
    \lambda(s,x)\,ds
  \end{equation}
  is a martingale. Moreover, $\lambda(t,x)$ is decreasing and continuous in
  $x$ with $ \lambda(t,1)=0$.

Consider a $d$-dimensional L\'evy process $Z$.
Denote by $\langle  \cdot, \cdot \rangle$ the scalar product in
$\R^d$, by $\cdot^\top$ the transpose and by $|\cdot|$ the
respective norm on $\R^d$.
It is well-known that a L\'evy-process can be decomposed as
\begin{align}\label{dec:Z}
Z_t &= mt+W_t + \int_0^t\int_{|z| \le 1} z \, ( \mu(ds,dz) -\nu(dz)ds)
 +\int_0^t \int_{| z | > 1} z \, \mu(ds,dz),
\end{align}
where $m \in \RR^d$, $W$ is a  $d$-dimensional Wiener process with covariance matrix $\Sigma$  and $\mu$ is the random measure of jumps
with its $\P$-compensator $\nu(dz)ds $.
That is, for any Borel set $B$ of $\RR^+$  and any Borel set $\Lambda$  of $\R^d$,
$\mu$ denotes the number of jumps in the time interval $B$ which have sizes in $\Lambda$,
$$\mu(\omega; B,\Lambda)=\sum_{s\in B}1_{\Lambda}(\Delta Z_s).$$
The process $Z$ has exponential moments  if the Laplace transform is always finite. Then the
Laplace transform satisfies
$\E(\, e^{ -\langle u,Z_t \rangle} ) = e^{t J(u)}$ for all $u\in \R^d$ with
\begin{align}\label{def:J}
J(u) = - \langle m ,u \rangle  +  \half \langle \Sigma u,u \rangle + \int_{\R^d} \Big( e^{-\langle u,z\rangle} -1 + \ind {|z| \le 1 }(z) \, \langle u,z \rangle \Big) \nu(dz).
\end{align}
We assume throughout  that the $(T,x)$-forward rate is given by
\begin{align}\label{feq}
\begin{split}  f(t,T,x)=f(0,T,x)&+ \int_0^t a(s,T,x)ds +\int_0^t \langle b(s,T,x), dZ_s\rangle  \\
  &+\int_0^t\int_\Ical c(s,T,x;y)\,\mu^L(ds,dy)
\end{split}\end{align}
 where $\mu^L$ is the
random measure associated to the jumps of $L$. 
 By $[L,Z]$ we denote the covariation of $L$ and $Z$. Let
$$ B:=\Big\{ u \in \R^d: \int_{|z|>1} e^{-\langle u,z \rangle} \nu(dz)  < \infty\Big\}. $$
Additionally, we use the following assumptions.
\begin{description}
\item[(A2)]  For each $T,x$ the processes $a,b$, and $c$ are assumed to be predictable and have
bounded trajectories. Furthermore, $c(\cdot,T,1;y)=0$ for all $y \in \cI$.
\item[(A3)] $[L,Z]_t=0$ for all $t \ge 0$.
\item[(A4)] For any $r >0$ the function $u \mapsto \int_{|z|>1} e^{-\langle u,z \rangle} \nu(dz)$ is bounded on $\{u \in \R^d:|u| \le r\} \cap B$.
\item[(A5)] For each $(T,x)$ there exists $K(T,x)$ such that
$$ \sup_{0 \le s \le T, y \in \R^d} c(s,T,x;y) \le K(T,x). $$
\end{description}
The assumption on $c$ in (A2) states that losses in the considered portfolio do not influence the risk-free rate.
This assumption can be relaxed but at the cost of further notation. Assumption (A3) is natural from a modelling point of view:
jumps in $L$ influence $f$ only through $c$ and not via a dependence of $L$ and $Z$.
\paragraph{Contagion.} This framework allows for two kinds of \emph{contagion}, i.e.\
feedback of the loss process $L$ on the forward rates:
first, a direct contagion via simultaneous jumps of $L$ and $f$; when
$\Delta L_t \not  =0$, (A3) gives that
 $\Delta f(t,T,x)=c(t,T,x;\Delta L_t)$.
Second, a kind of indirect contagion via  letting the model parameters $a$, $ b$,
  and $c$ be explicit functions of the  loss path $L$. \\[-2mm]

It is well-known that then the market of $(T,x)$-bonds is free of arbitrage if
\begin{align}\label{EMM}
(D_{t} P(t,T,x))_{0 \le t \le T} \quad \text{are local martingales for all } (T,x),
\end{align}
where $D$ is the discounting process given by
\begin{align*} D_{t} = e^{-\int_0^t r_s ds}= e^{-\int_{0}^t f(s,s,1)ds}.\end{align*}
The goal of this section is to give conditions which are sufficient
for \eqref{EMM} to hold.

\paragraph{Default-free market.}
The default-free forward rate, $f(t,T)$, is given by $$f(t,T)=f(t,T,1).$$ We also denote $a(t,T)=a(t,T,1)$ and $b(t,T)=b(t,T,1)$.
In the case of default free markets, the following was
shown in
\citeN{JakubowskiZabczyk:EMHJM}, Theorem 3.1.

\noindent On one side, (A2) and the no-arbitrage condition
\begin{align}\label{EMM:risk-free}
D_{t} P(t,T) \quad \text{are local martingales for all } 0 \le t \le
T;
\end{align}
imply that for any $u \le T$
\begin{align}\label{Bcondition}
\int_t^u b(t,v) dv \in B \qquad \text{for almost all } t \in [0,u].
\end{align}
On the other side, if (A2), (A3) and \eqref{Bcondition} hold, then \eqref{EMM:risk-free} is equivalent to
\begin{align}\label{driftcondition-riskfree}
    \int_t^s a(t,u) du  = J\Big( \int_t^s b(t,u) du\Big)
\end{align}
for almost all $0 \le t \le s \le T $.

\paragraph{The drift conditions.}
Now we are in the position to state the model for the defaultable market and derive the drift conditions.
Recall that we consider a market consisting of $(T,x)$-bonds.

\begin{thm} \label{thm1}
Assume that (A1)-(A4) hold.
\begin{enumerate}\renewcommand{\labelenumi}{(\roman{enumi})\ }
\item If (A5) holds, it follows from \eqref{EMM}  that
\begin{align} \label{Bxcondition}
\int_t^s b(t,v,x) dv \in B
\end{align}
for any $0 \le t \le s$ on $\{L_t \le x\}$, $\Q\otimes dt$-a.s. \\
\item If \eqref{Bxcondition} holds, then \eqref{EMM} is equivalent to
\begin{align}
    \int_t^s a(t,u,x) du  & = J\Bigg( \int_t^s b(t,u,x) du\Bigg)   \nonumber \\
    & + \int_\Ical\left(e^{-\int_t^s c(t,u,x;y)\,du}-1\right)\ind{ L_t+y\le x}\nu^L(t,dy)
 \label{dc1}\\
    f(t,t,x) &= f(t,t)+\lambda(t,x) \label{dc2}
\end{align}
for any $0 \le t \le s$, where \eqref{dc1} and \eqref{dc2} hold on $\{L_t \le x\}$, $\Q\otimes dt$-a.s.
\end{enumerate}

\end{thm}

\begin{proof}
Set
\begin{align*}
  p(t,T,x) &:=\exp \Bigg( - \int_t^{T} f(t,u,x) \, du \Bigg),
\end{align*}
such that $P(t,T,x) = \ind{L_t\le  x} p(t,T,x)$. Recall
the martingale $M^x$ from \eqref{eqMx}. Then
\begin{align}
  d \big( D_t P(t,T,x)\big)   
              &= D_t p(t-,T,x) dM^x_t - D_t p(t-,T,x) \lambda(t,x) \ind{L_{t}\le x} \, dt  \nonumber\\
                            &+ \ind{L_{t-}\le x } d(D_t p(t,T,x))
                            + d [\ind{L_t \le x}, D_t p(t,T,x)]  . \label{temp644}
\end{align}
Denote  $a^*(t,T,x) := \int_t^T a(t,u,x) du$, $b^*(t,T,x) := \int_t^T b(t,u,x) du$,  and similarly $c^*(t,T,x;y) := \int_t^T c(t,u,x;y) du$.
We first compute the last term in \eqref{temp644}. Note that the process $(\ind{L_t\le x})_{t \ge 0}$ jumps at most  once, from $1$ to $0$, at the first time when $L$ crosses
the barrier $x$.
 Hence,
\begin{align}
\lefteqn{[\ind{L_\cdot \le x}, D_\cdot p(\cdot,T,x)]_t = \sum_{0 \le s \le t} \Delta  \ind{L_s \le x} \, \Delta (D_s p(s,T,x)) }\qquad\quad \nonumber\\
&= - \sum_{0 \le s \le t} \ind{ L_{s-} \le x, L_s > x} \, \Delta (D_s p(s,T,x)) \nonumber\\
&= - \int_0^t  \int_\cI  \ind{L_{s-} \le x} D_s  p(s-,T,x)  \ind{ L_{s-}+y > x}  \left(e^{-c^*(s,T,x;y)}-1 \right) \mu^L(ds,dy),
\label{temp514}
\end{align}
where we used in the last step that by (A3), $Z$ and $L$ do not have simultaneous jumps.
Regarding the remaining term in \eqref{temp644},   we obtain
from \eqref{eq:TxbondsViaForwardRates} and \eqref{feq}, by the \Ito-formula, that
\begin{align}
\ind{L_{t-}\le x} d \big(D_tp(t,T,x) \big) \nonumber
& =  D_{t} p(t-,T,x) \ind{L_{t-}\le x} \\
&\Bigg(\Big(f(t,t,x)- r_t -  \langle b^*(t,T,x),m\rangle -a^*(t,T,x)\Big) dt \nonumber\\[2mm]
&+ \half  b^*(t,T,x)^\top \, \Sigma \, b^*(t,T,x) dt  \nonumber\\[2mm]
& +  \int _{ \R^d} \Big[ e^{-  b^*(t,T,x)^\top  z } -1+b^*(t,T,x)^\top z \Big]\ind{|z| \le 1} \,  \mu(dt,dz)  \nonumber\\
& +  \int _{ \R^d} \Big[ e^{-  b^*(t,T,x)^\top  z } -1 \Big]\ind{|z| > 1} \,  \mu(dt,dz)  \nonumber\\
& +  \int _{ \cI} \ind{L_{t-}\le x} \Big[ e^{-  c^*(t,T,x;y)  } -1\Big] \,   \mu^L(dt,dy)  \Bigg)\nonumber \\
& + d\tilde M_t \label{eq:dynamicsDp}
  \end{align}
where $\tilde M$ is a local martingale and we used the decomposition \eqref{dec:Z}.
Inserting \eqref{temp514} and \eqref{eq:dynamicsDp} in \eqref{temp644} we obtain, using  $\ind{L_{t-}\le x} p(t-,T,x) = P(t-,T,x)$, that
\begin{align}
d\big(  D_tP(t,T,x)\big) &=  D_t P(t-,T,x) \nonumber \\
&\Bigg(  dM^x_t -   \lambda(t,x) \ind{L_{t}\le x} \, dt  \nonumber\\
&\Big(f(t,t,x)- r_t -  \langle b^*(t,T,x),m\rangle -a^*(t,T,x)\Big) dt \nonumber\\[2mm]
&+ \half  b^*(t,T,x)^\top \, \Sigma \, b^*(t,T,x) dt  \nonumber\\[2mm]
& +  \int _{ \R^d} \Big[ e^{-  b^*(t,T,x)^\top  z } -1+b^*(t,T,x)^\top z \Big]\ind{|z| \le 1} \,  \mu(dt,dz)  \nonumber\\
& +  \int _{ \R^d} \Big[ e^{-  b^*(t,T,x)^\top  z } -1 \Big]\ind{|z| > 1} \,  \mu(dt,dz)  \nonumber\\
& +  \int _{ \cI}  \ind{L_{t-}+y \le x}\Big( e^{-  c^*(t,T,x;y)  } -1\Big) \,   \mu^L(dt,dy)  \Bigg)+d \tilde{\tilde M}_t \label{eq:dynamicsDP}
  \end{align}
  with a local martingale $\tilde{\tilde M}$.  First we study (i). From \eqref{eq:dynamicsDP} we obtain that
\begin{align}
 D_tP(t,T,x) &= D_0P(0,T,x)  + A(t)+ {\tilde{\tilde {\tilde M}}}_{t}  \nonumber \\
&+ \int_0^t  \int _{ \R^d}D_{s}P(s-,T,x) \Big( e^{-  b^*(s,T,x)^\top  z } -1+b^*(s,T,x)^\top z \Big) \ind{|z| \le 1} \,  \mu(ds,dz)  \nonumber\\
&+ \int_0^t  \int _{ \R^d} D_{s}P(s-,T,x)\Big( e^{-  b^*(s,T,x)^\top  z } \Big) \ind{|z| > 1}  \,  \mu(ds,dz)  \nonumber\\
&- \int_0^t  \int _{ \R^d} D_{s}P(s-,T,x)\ind{|z| > 1}  \,  \mu(ds,dz)  \nonumber\\
& +  \int_0^t \int _{ \cI}D_{s}P(s-,T,x)  \ind{L_{s-}+y \le x}\Big( e^{-  c^*(s,T,x;y)  } -1\Big) \,   \mu^L(ds,dy)  \Bigg) \nonumber\\
&=: D_0P(0,T,x)  + A(t)+ {\tilde{\tilde {\tilde M}}}_{t}+ I_1(t) +  I_2(t) - I_3(t) + I_4(t). \label{def:I}
\end{align}
Here ${\tilde{\tilde {\tilde M}}}_{t}$ is a local martingale and $A(t)$ is a $dt$ integral and thus they both are locally integrable processes. The no-arbitrage condition \eqref{EMM} implies that the sum of all those six processes is a local martingale and in particular  the sum
$$
I_1(t) +  I_2(t) - I_3(t) + I_4(t)
$$ 
is locally integrable.
We consider each term separately.
Since $(D_tP(t,T,x))_{0 \le t \le T}$ is a nonnegative local martingale, it is a supermartingale and hence
\begin{align}
\label{intDP}
\sup_{0 \le s \le T} \E(D_sP(s-,T,x))\leq \sup_{0 \le s \le T} \E(D_sP(s,T,x)) \le \E(D_0P(0,T,x)) < \infty
\end{align}
the last expectation being finite by \eqref{EMM}. The process  
$$
I_ 3(t) :=  \int_0^t  \int _{ \R^d} D_s P(s-,T,x) \ind{|z| > 1}  \,  \mu(ds,dz)
$$ 
is  locally integrable. To see this let  $(\tau_n)$ be an increasing sequence of stopping times dominated by $T$.    Then
\begin{align}\label{temp574}
   \E \bigg( \int_0^{\tau_n} \int_{\R^d}  & D_s P(s-,T,x) \ind{|z| > 1}  \mu(ds,dz) \bigg)      \\
&\leq  \E \left( \int_0^T \int_{\R^d} D_s P(s-,T,x) \ind{|z| > 1}  \, \mu(ds,dz) \right) \nonumber\\
&=   \int_0^T \int_{\R^d} \E\left( D_s P(s-,T,x)
 \right)\ind{|z| > 1}  \nu(dz) ds.
\nonumber
\end{align}
Since the integrands are nonnegative, we obtain with \eqref{intDP} that
\begin{align*}
\E \bigg( \int_0^{\tau_n} \int_{\R^d}  & D_s P(s-,T,x) \ind{|z| > 1}  \mu(ds,dz) \bigg) & \le T  \E(D_0P(0,T,x)) \nu(\{z \in \R^d: |z|>1\}) < \infty. \end{align*}
Next, we consider $I_4$.  
By (A5)
$$ \Big| e^{-c^*(s,T,x;y)}Ê- 1 \Big| \le e^{T K(T,x)} + 1 =: \tilde K(T,x)$$
and hence
\begin{align*}
&\E\left( \int_0^{\tau_n} \int_{\R^d}  D(s) P(s-,T,x) \ind{L_{s-}+y \le x}\Big| e^{-  c^*(s,T,x;y)  } -1\Big| \,   \mu^L(ds,dy)  \right) \\
& \qquad\le \tilde K(T,x) \E \left( \int_0^T \int_{\R^d} \Ind_{[0,\tau_n]}(s) D(s) P(s-,T,x) \ind{L_{s-}+y \le x}   \mu^L(ds,dy)  \right) \\
& \qquad\le \tilde K(T,x)\int_0^T \int_{\R^d} \E \left(   D(s) P(s-,T,x)     \right) \nu^L(s,dy),
\end{align*}
which is finite by \eqref{intDP}, (A1) and the second part in (A5).
Summarizing, the no-arbitrage condition \eqref{EMM} implies that the sum $I_1+I_2 $ is locally integrable.
However
\begin{align*}
I_1 + I_2  &=
 \int_0^t  \int _{ \R^d} \Big( e^{-  b^*(s,T,x)^\top  z } -1+b^*(s,T,x)^\top z \Big) \ind{|z| \le 1} \,  \mu(ds,dz)  \nonumber\\
&+ \int_0^t  \int _{ \R^d} e^{-  b^*(s,T,x)^\top  z }  \ind{|z| > 1}  \,  \mu(ds,dz).
\end{align*}
Since both integrands are nonnegative, this sum is locally  integrable if and only if both summands are locally integrable and
so $\int_0^t  \int _{ \R^d} e^{-  b^*(s,T,x)^\top  z }  \ind{|z| > 1}  \,  \mu(ds,dz)$ is locally integrable. 
Therefore, for a localizing sequence $\tau_n$ 
$$
\E \int_0^T  \int _{ \R^d} e^{-  b^*(s,T,x)^\top  z }\ind{s\leq \tau_n}\ind{|z| > 1}\mu (ds,dz) < +\infty .  
$$
Equivalently
$$
 \E \Big[ \int_0^{\tau_{n}} \big( \int _{|z| >1}  e^{-  b^*(s,T,x)^\top z }\nu (dz)\big)ds\Big]  < +\infty.
$$
This implies  \eqref{Bxcondition} and assertion (i) follows.

We consider now (ii).  Note that $D_t p(t-,T,x) >0$.
Compensating   in (\ref{eq:dynamicsDp}) the integrals
with respect to the random measures $\mu$ and $\mu^L$,  collecting
the 'dt'-terms and dividing by $D_t p(t-,T,x)>0$ gives
 \begin{align}\ind{L_t \le x}\Bigg[
&    f(t,t,x)- r_t -\lambda(t,x)
-  \langle b^*(t,T,x),m\rangle -a^*(t,T,x) \nonumber\\
&+ \half  b^*(t,T,x)^\top \, \Sigma \, b^*(t,T,x)  \nonumber\\[2mm]
& +  \int _{ \R^d} \Big[ e^{-  b^*(t,T,x)^\top  z } -1+\ind{|z| \le 1}b^*(t,T,x)^\top z \Big] \,  \nu(dz) \nonumber\\
& +  \int _{ \cI} \ind{L_{t}+ y \le x }\Big( e^{-  c^*(t,T,x;y)  } -1\Big)  \,  \nu^L(t,dy) \Bigg] dt. \label{eq:drift}
\end{align}
$Dp$ being a local martingale is equivalent to having a vanishing drift.
From \eqref{eq:drift}  we  obtain the following   condition:
\begin{align}
0
&=     f(t,t,x)- r_t -\lambda(t,x) -  a^*(t,T,x)  + J(b^*(t,T,x)) \nonumber\\
& +  \int _{ \cI} \Big( e^{-  c^*(t,T,x;y)  } -1\Big) \ind{L_{t}+ y \le x } \,  \nu^L(t,dy), \label{temp586}
\end{align}
$dt \otimes \Q$-almost surely for all $T \ge t$. Let $N\in \cB(\R^+)\otimes \cF$ be the  set of all $(t,\omega)$ for which \eqref{temp586} holds.
For $(t,\omega)\in N$ we choose  $T= t$ and hence  obtain \eqref{dc2} as $J(0)=0$.
The remaining terms give \eqref{dc1}, where both
equations hold only on $N$, i.e.\ $dt \otimes \Q$-almost surely.

For the converse, we need to show that  the drift conditions imply that all discounted $(T,x)$-bonds are local martingales.
For fixed $x$, the drift conditions imply that the 'dt'-terms in \eqref{eq:dynamicsDP} vanish (compare \eqref{eq:drift})
and hence $(D_tP(t,T,x))_{0 \le t \le T}$  are local martingales. The conclusion follows.
\end{proof}


Theorem~\ref{thm1} states that, under the no-arbitrage
condition~\eqref{EMM}, the {\em drift}
$a(t,\cdot,x)$ of the forward rates is determined by the {\em volatility}
$b(t,\cdot,x)$ and the \emph{contagion} $c(t,\cdot,x;\cdot)$. Besides this, \eqref{dc1} gives an
implicit relation between $f(t,t,x)$ and  $L_t$ or more precisely its compensator;
the relationship between $\lambda$ and the compensator $\nu$ is given in \eqref{deflambda}.

In \citeN{FilipovicSchmidtOverbeck} it is shown how to construct a conditional
Markov process which satisfies the drift conditions when the filtration is generated by a Brownian motion.
The extension to the case where the background filtration is generated by a
L\'evy process is studied in \citeN{TappeSchmidt}.

There is a rich literature on drift conditions in the defaultable and default-free case, see for example \citeN{BMKR},  \citeN{MusielaRutkowski}, \citeN{Filipovic},
 \citeN{TSchmidt_InfiniteFactors}, \citeN{SchmidtOezkan} and \citeN{Huehne}.

\section{Market Models}

For practical applications it is important to realize, that $(T,x)$-bond prices are neither available for all
maturities nor for all levels $x$. This chapter is devoted to the study of a
\emph{market model} where this assumption is relaxed. We start by deriving dynamics of $(T_k,x_i)$-rates in
section 4.1, and consider market models of the forward rate and the $(T_k,x_i)$-rate in sections 4.2 and 4.3, respectively.
The drift condition derived in Theorem \ref{thm1} is, however,  still sufficient for absence of arbitrage.
Recall that $(L_t)_{t \ge 0}$ was the loss process of the CDO and $P(t,T,x)=\ind{L_t \le x} p(t,T,x)$ with  $p(t,T,x)>0$.

We fix a tenor structure $0 < T_1 <\dots< T_n$  and a barrier structure $0 =x_1 < \dots < x_m = 1$. Set $\cS:=\{T_1,\dots,T_n\} \times \{x_0,\dots,x_m\}$. The considered market model consists of all $(T,x)$-bonds with $(T,x)  \in \cS$.
As $x_m=1$, the market also contains risk-free bonds. Throughout we assume that the initial term structure, denoted in $(T,x)$-bonds, is strictly positive and decreasing in $T$ and increasing in  $x$.

\begin{defin}  Set $\delta_k := T_{k+1}-T_k$. The rate
$$ L(t,T_k,x_i) := \ind{L_t \le x_i}\frac{1}{\delta_k} \left( \frac{p(t,T_k,x_i)}{p(t,T_{k+1},x_i)} -1 \right)$$
is called \emph{$(T_k,x_i)$-rate}.
For all $0\le t\le S\le T$ and $x \in \cI$,
\begin{align}\label{def:F}
F(t,S,T,x) := \ind{L_t \le x} \frac{p(t,S,x)}{p(t,T,x)}
\end{align}
defines the \emph{$(S,T,x)$-forward bond price}.
\end{defin}
The $(T_k,1)$-rate is the so-called LIBOR-rate. It  is a default-free
and has been studied in many papers, see e.g.~\citeN{FilipovicBook}. Defaultable LIBOR-rates have been studied in
\citeN{EberleinKlugeSchoenbucher06}. In this paper we follow \citeN{BGM} and embed the $(T_k,x_i)$-rates in the framework of continuous forward rates
and obtain sufficient conditions for absence of arbitrage.
In \citeN{GrbacEberleinSchmidt}  $(T_k,x)$-rates are studied directly; however they do not consider a discrete structure for the
loss levels.

\subsection{Dynamics of $(T_k,x_i)$-rates}
In this section we derive the dynamics of $(T_k,x_i)$-rates when the
drift condition \eqref{dc1} is satisfied. This result generalizes \citeN{eberlein.oezkan05} and \citeN{Huehne}.
Recall that $b^*(t,T,x) = \int_t^T b(t,u,x) du$.  
\begin{thm}
Assume that (A1)-(A4) and   \eqref{dc1} hold and consider $(T_k,x_i)\in \cS$. On $\{L_{t-} \le x_i\}$ we have that
\begin{align}\label{MM:L}
dL(t,T_k,x_i) &=
\frac{1+\delta_k L(t-,T_k,x_i)}{\delta_k}\bigg(
  D(t,T_k,T_{k+1},x_i) -\frac{\delta_k L(t-,T_k,x_i)}{1+\delta_k L(t-,T_k,x_i)}\lambda(t,x_i)\bigg) dt \nonumber \\
&+d\tilde M_t,
\end{align}
where the drift term attributed to the compensated jumps and the quadratic variation of the diffusive part is
\begin{align}
\lefteqn{D(t,T_k,T_{k+1},x_i)} \nonumber\\
 &:= \int_{\R^d} \Big( e^{ \langle b^*(t,T_{k+1},x_i) - b^*(t,T_{k},x_i), z \rangle  }- e^{- \langle b^*(t,T_k,x_i) , z \rangle}
 + e^{ -\langle b^*(t,T_{k+1},x_i) , z \rangle} -1 \Big) \nu(dz) \nonumber \\
&+ \int_\cI \Big(e^{c^*(t,T_{k+1},x_i;y)-c^*(t,T_k,x_i;y)}  - e^{-c^*(t,T_k,x_i;y)}+e^{-c^*(t,T_{k+1},x_i;y)}\Big) \ind{L_{t}+y \le x_i}\nu^L(t,dy) \nonumber\\
&+ \langle  \Sigma (b^*(t,T_{k+1},x_i) - b^*(t,T_k,x_i)), b^*(t,T_{k+1},x_i) \rangle
\label{def:DriftD}
\end{align} and
 $\tilde M=\tilde M(T_k,x_i)$ is the  local martingale
\begin{align} \label{MM:L2}
d\tilde M_t &:=  L(t-,T_k,x_i)  dM_t^x \nonumber\\
&+
\frac{1+\delta_k L(t-,T_k,x_i)}{\delta_k} \Bigg(\int_{\R^d} \left( e^{ \langle b^*(t,T_{k+1},x_i) - b^*(t,T_{k},x_i), z \rangle } -1  \right) (\mu(ds,dz)-\nu(dz)ds) \nonumber\\
&+\int_{\cI}\left( e^{c^*(t,T_{k+1},x_i;y)-c^*(t,T_k,x_i;y) } -1  \right)  \ind{L_{t-}+y\le x_i} (\mu^L(ds,dy)-\nu^L(dy)ds) \nonumber\\[2mm]
&+  \langle b^*(t,T_{k+1},x_i) - b^*(t,T_{k},x_i),  dW(t) \rangle \Bigg) .
\end{align}

\end{thm}
\begin{proof}
Fix $0<S<T$. First, we derive the dynamics of the pre-default $(S,T,x)$-forward bond price process
\begin{align} \label{def:g}
 g(t,S,T,x) := \frac{p(t,S,x)}{p(t,T,x)}.
\end{align}
With $x,S$ and $T$  fixed we denote   $A(t):= \int_S^T a(t,u,x) du,$  $B(t) := \int_S^T b(t,u,x) du$, and $C(t;y):=\int_S^T c(t,u,x;y) du$ .
By \eqref{def:g} and the dynamics of the forward rate $f$ given  in  \eqref{feq}, the stochastic Fubini theorem yields
\begin{align*}
g(t,S,T,x)
 &= e^{\int_S^T \left( f(0,u,x) + \int_0^t a(s,u,x) ds + \int_0^t \langle b(s,u,x), dZ_s \rangle +\int_0^t\int_\cI c(s,u,x;y) \mu^L(ds,dy)\right) du} \\
 &= g(0,S,T,x) \, e^{\int_0^t A(s) ds + \int_0^t \langle B(s), dZ(s) \rangle + \int_0^t\int_\cI C(s;y) \mu^L(ds,dy)}.
\end{align*}
The It\^o-formula gives an expression for the dynamics of $g$:
\begin{align}
dg(t,S,T,x) &=   g(t-,S,T,x) \Big( A(t) + \langle B(t), dZ(t) \rangle \Big)dt \nonumber\\
   &+ \half g(t-,S,T,x)  \langle \Sigma B(t), B(t) \rangle  dt \nonumber \\
& +  \Big( g(t,S,T,x) - g(t-,S,T,x) - g(t-,S,T,x) \langle B(t), \Delta Z(t) \rangle \Big) \nonumber \\[2mm]
&=: d I_1(t) + dI_2(t) + dI_3(t). \label{def:I123}
\end{align}
The decomposition of $Z$ in \eqref{dec:Z} yields
\begin{align*}
dI_1(t) &=   g(t-,S,T,x) \left( A(t) + \langle B(t), m \rangle \right) dt \\[2mm]
& + g(t-,S,T,x)\langle B(t), dW(t) \rangle \\
& +\int_{\parallel z \parallel \le 1} g(t-,S,T,x) \langle B(t) , z \rangle (\mu(dt,dz) - \nu(dz)dt) \\
& +  \int_{\parallel z \parallel > 1} g(t-,S,T,x) \langle B(t) , z \rangle \mu(dt,dz).
\end{align*}
By (A3), $L$ and $Z$ have no joint jumps and therefore
\begin{align*}
dI_3(t) &=  g(t-,S,T,x) \left( e^{ \langle B(t), \Delta Z(t) \rangle + C(t,\Delta L(t)) } -1 \right)
- g(t-,S,T,x) \langle B(t), \Delta Z(t) \rangle \\
&=  \int_{\R^d} g(t-,S,T,x) \left( e^{ \langle B(t), z \rangle } -1 - \langle B(t), z \rangle \right) \mu(dt,dz) \\
&+ \int_\cI g(t-,S,T,x) \left( e^{ C(t;y)} -1 \right) \mu^L(dt,dy),
\end{align*}

Inserting this expressions in \eqref{def:I123} we obtain
\begin{align}
\frac{dg(t,S,T,x)}{g(t-,S,T,x) } &=    \left( A(t) + \langle B(t), m \rangle + \half  \langle \Sigma B(t), B(t) \rangle\right) dt  +   \langle B(t), dW(t) \rangle \nonumber \\
&+\int_{\R^d}  \left( e^{ \langle B(t), z \rangle } -1  \right) (\mu(dt,dz) -\nu(dz)dt) \nonumber \\
&+ \int_{\R^d}  \left( e^{ \langle B(t), z \rangle } -1 - \ind{\parallel z \parallel \le 1}\langle B(t),z \rangle \right) \nu(dz)dt   \\
& +\int_{\cI} \left( e^{ C(t;y) } -1  \right) \mu^L(dt,dy)
. \label{temp221}
\end{align}
As $L(t,T_k,x) = \ind{L_t \le x}\delta_k^{-1}( g(t,T_k,T_{k+1},x)-1)$,
we have that
\begin{align*}
d L(t,T_k,x) &= \delta_k^{-1}( g(t-,T_k,T_{k+1},x)-1) d \ind{L_t\le x}  \\
&+ \ind{L_{t-} \le x}  \delta^{-1}_k dg(t,T_k,T_{k+1},x) \\
&+ d [\ind{L_t \le x}, \delta_k^{-1}( g(t,T_k,T_{k+1},x)-1) ]. \end{align*}
Note that $d\ind{L_t \le x}=dM_t^x - \ind{L_t \le x} \lambda(t,x) dt$ and,
similar to \eqref{temp514},
\begin{align*}
\lefteqn{d [\ind{L_t \le x}, \delta_k^{-1}( g(t,T_k,T_{k+1},x)-1) ] } \qquad \qquad\\
&=
\delta_k^{-1}\int_{\cI} \big( \ind{L_{t-}+y  \le x} - \ind{L_{t-} \le x}\big) g(t-,T_k,T_{k+1},x)\left( e^{C(t,T_k,T_{k+1};y)}-1\right)\mu^L(dt,dy).
\end{align*}
Summarizing, we obtain the following dynamics of $L$,
\begin{align}
{d L(t,T_k,x) }&= {L(t-,T_k,x)} \big( dM_t^x - \lambda(t,x) dt \big) \nonumber\\
&+ \ind{L_{t-} \le x} \frac{\delta_k L(t-,T_k,x) +1}{\delta_k} \bigg[
\Big( A(t) + \langle B(t),m \rangle + \half \langle \Sigma B(t), B(t) \rangle  \nonumber\\
&+ \int_{\R^d} \Big( e^{\langle B(t),z \rangle} -1 -  \ind{\parallel z \parallel  \le 1} \langle B(t),z\rangle \Big) \nu(dz) \Big) dt \nonumber\\
&+\langle B(t), dW(t) \rangle+\int_{\R^d} \Big( e^{\langle B(t),z\rangle}-1 \Big) (\mu(dt,dz) - \nu(dz)dt) \nonumber\\
&+\int_\cI \Big( e^{C(t;y)}-1 \Big) \mu^L(dy,dt) \nonumber\\
& + \int_{\cI} \big( \ind{L_{t-}+y  \le x} - \ind{L_{t-} \le x}\big) \left( e^{C(t;y)}-1\right)\mu^L(dt,dy)\bigg].
\label{eq:dynL}
\end{align}
Compensating the remaining $\mu^L$ terms and
collection all $ds$-terms and using \eqref{def:J} gives the drift term of $L$,
\begin{align}\label{eq:driftL}
&- L(t-,T_k,x) \lambda(t,x) \\
&+ \ind{L_{t-} \le x} \frac{\delta_k L(t-,T_k,x) +1}{\delta_k} \bigg[ A(t) + J(-B(t)) + \int_{\cI}  \left( e^{ C(t;y) } -1  \right)\ind{L_{t}+y \le x}\nu^L(t,dy) \bigg]. \nonumber
\end{align}
In the next step we will use the drift condition \eqref{dc1} to work further on this expression. Recall that
\begin{align*}
B(t) &= \int_S^T b(t,u,x) du = b^*(t,T,x) - b^*(t,S,x),
\end{align*}such that  \eqref{def:J} gives
\begin{align*}
 J(-B(t) ) &= -\langle m, b^*(t,S,x) \rangle +\langle m, b^*(t,T,x) \rangle  \\
 & + \half \langle  \Sigma (b^*(t,T,x)-b^*(t,S,x)), (b^*(t,S,x) - b^*(t,T,x))\rangle \\[2mm]
 & + \int_{\R^d} \left(e^{- \langle b^*(t,S,x) - b^*(t,T,x), z \rangle  } -1 +
 \ind{\parallel z \parallel \le 1}\langle b^*(t,S,x) - b^*(t,T,x), z \rangle  \right)\nu(dz)
 \\ &= J\big(b^*(t,S,x)\big) -J \big( b^*(t,T,x) \big)  \\[2mm]
 &+ \langle  \Sigma (b^*(t,T,x)-b^*(t,S,x) ), b^*(t,T,x) \rangle \\
 &+ \int_{\R^d} \left( e^{- \langle b^*(t,S,x) - b^*(t,T,x), z \rangle  } - e^{- \langle b^*(t,S,x) , z \rangle}
 + e^{ -\langle b^*(t,T,x) , z \rangle} -1 \right) \nu(dz)
\end{align*}
With $A(t) = a^*(t,T,x) - a^*(t,S,x)$
the drift condition \eqref{dc1} gives
\begin{align*}
A(t) &= J(b^*(t,T,x)) - J(b^*(t,S,x)) \\
&+\int_\cI \Big( e^{-c^*(t,T,x;y)}- e^{-c^*(t,S,x;y)}\Big) \ind{L_{t}+y \le x}\nu^L(t,dy).
\end{align*}Hence,
\begin{align}
\lefteqn{A(t) + J(-B(t)) +  \int_{\cI}  \left( e^{ C(s;y) } -1  \right)\ind{L_{t}+ y \le x}\nu^L(t,dy) }\quad\nonumber\\
&=  \langle  \Sigma  (b^*(t,T,x)-b^*(t,S,x)) , b^*(t,T,x) \rangle \nonumber\\
 &+ \int_{\R^d} \left( e^{ \langle b^*(t,T,x) - b^*(t,S,x), z \rangle  } - e^{- \langle b^*(t,S,x) , z \rangle}
 + e^{ -\langle b^*(t,T,x) , z \rangle} -1 \right) \nu(dz) \nonumber \\
&+\int_\cI \Big(e^{c^*(t,T,x;y)-c^*(t,S,x;y)} - e^{-c^*(t,S,x;y)}+e^{-c^*(t,T,x;y)}\Big) \ind{L_{t}+y \le x}\nu^L(t,dy).\label{temp222}
\end{align}
Inserting this in \eqref{eq:driftL} and using  the dynamics of $L$ in \eqref{eq:dynL} we obtain the result.
\end{proof}

As byproduct of the above proof we obtain the dynamics of the $(S,T,x)$-forward rate.
\begin{col}\label{lem:F}
Assume (A1)-(A4) and \eqref{dc1} holds and $(T_k,x_i)\in \cS$ with
$k < n$.  Then, on $\{L_{t- \le x_i}\}$,
\begin{align}\label{MM:F}
\frac{dF(t,T_k,T_{k+1},x_i)}{F(t-,T_k,T_{k+1},x_i)} &=  \Big( - \lambda(t,x_i)
 +D(t,T_k,T_{k+1},x_i) \Big) dt  + d\tilde{\tilde M}_t
\end{align}
with $D$ from \eqref{def:DriftD} and the local martingale $\tilde{\tilde M}$ given in \eqref{def:ttM}.
\end{col}
\begin{proof}
Noting that $F(t,T_k,T_{k+1},x_i) = \ind{L_t \le x_i} g(t,T_k,T_{k+1},x_i)$ with $g$ defined in \eqref{def:g}, we obtain the dynamics of $F$ via
\begin{align*}
d F(t,T_k,x_i) &=  g(t-,T_k,T_{k+1},x_i) d \ind{L_t\le x_i}  \\
&+ \ind{L_{t-} \le x_i}   dg(t,T_k,T_{k+1},x_i) \\
&+ d [\ind{L_t \le x_i},  g(t,T_k,T_{k+1},x_i) ]. \end{align*}
As above we have that
\begin{align*}
\lefteqn{
d [\ind{L_t \le x_i},  g(t,T_k,T_{k+1},x_i ] } \hspace{0.5cm}\\ &=
\int_{\cI} \big( \ind{L_{t-}+y  \le x_i} - \ind{L_{t-} \le x_i}\big) g(t-,T_k,T_{k+1},x_i)\left( e^{C(t,T_k,T_{k+1};y)}-1\right)\mu^L(dt,dy).
\end{align*}
With \eqref{temp221} and \eqref{temp222} we obtain the stated dynamics. Here
\begin{align}\label{def:ttM}
d\tilde{\tilde M}_t &:=
  dM_t^x + \int_{\R^d} \left( e^{ \langle b^*(t,T_{k+1},x_i) - b^*(t,T_{k},x_i), z \rangle } -1  \right) (\mu(dt,dz)-\nu(dz)dt) \nonumber\\
&+\int_{\cI}\left( e^{(c^*(t,T_{k+1},x_i;y)-c^*(t,T_k,x_i;y)) } -1  \right)  \ind{L_{t-}+y\le x_i} (\mu^L(dt,dy)-\nu^L(dy)dt) \nonumber\\[2mm]
& + \langle b^*(t,T_{k+1},x_i) - b^*(t,T_{k},x_i),  dW(t) \rangle.
\end{align}
\end{proof}

\subsection{Forward rate modelling}
In practical applications one needs a simple structure of the $(T_k,x_i)$-rates which is analysed in this section. For notational convenience, we consider  forward bond prices instead of $(T_k,x_i)$-rates themselves. A result for the $(T_k,x_i)$-rate can be obtained in a similar way.
Recall from \eqref{eqMx} that $\lambda(t,x)$ was the intensity that $\ind{L_t\le x}$ jumps to zero at $t$.
Let $\eta(t):= \inf\{1 \le i \le n: T_{i+1}>t\}.$
In the following model the forward rate is driven  by the L\'evy process $Z$ through the
functions $\beta_k$; the reaction on the loss process $L$ can be adjusted through the functions
$\gamma_k$. These may depend on the loss occurring at $t$, $\Delta L_t$, and the current loss level $L_{t-}$ itself.
We need the following assumption:
\begin{description}
\item[(A6)]
For each $(T_k,x_i) \in \cS$ the  functions $\beta_{ki}:\R^+\to \R^d$,
$\gamma_{ki}:\R^+\times\cI \times [0,1] \to \R$ are measurable and bounded.
\end{description}
\begin{prop} Assume (A1)-(A6) hold.  Forward bond prices given on $\{L_{t}\le x_i\}$ by \label{thm:marketF}
\begin{align}\label{eq:marketF}
&\frac{dF(t,T_k,T_{k+1},x_i)}{F(t-,T_k,T_{k+1},x_i)}  =  \alpha_{ki}(t) dt + \langle \beta_{ki}(t),   dW(t) \rangle \nonumber\\
&\quad+ \int_{\R^d}  \left(e^{ \langle \beta_{ki}(t), z \rangle } -1  \right) \mu(dt,dz)+ \int_{\cI}  \left(e^{ \gamma_{ki}(t,L_{t-};y)} -1  \right)\ind{L_{t-}+y \le x_i} \mu^L(dt,dy),
\end{align}
$(T_k,x_i) \in \cS$, $k<n$ and zero on $\{L_{t}> x_i\}$ constitute  an arbitrage-free market if
\begin{align*}
\alpha_{ki} (t) &=  - \lambda(t,x_i) +  \sum_{j=\eta(t)}^{k} \langle \beta_{ji} (t), \Sigma \beta_{ki} (t) \rangle \\
&+ \int_{\R^d} \bigg(e^{\langle \beta_{ki}(t),z\rangle} +\Big(e^{-\langle \beta_{ki}(t),z\rangle}-1\Big) \prod_{j=\eta(t)}^{k-1}e^{ -\langle \beta_{ji}(t), z \rangle}  -1 \bigg) \nu(dz) \\
&+ \int_{\cI} \bigg(e^{\gamma_{ki}(t,L_{t-};y)} +\Big(e^{-\gamma_{ki}(t,L_{t-};y)}-1\Big) \prod_{j=\eta(t)}^{k-1}e^{- \gamma_{ji}(t,L_{t-};y)}  \bigg) \ind{L_t+y \le x_i} \nu^L(t,dy)
\end{align*}
for all $(T_k,x_i),(T_{k+1},x_i)\in \cS$.
\end{prop}
In this way the $(T,x)$-forward bond price explicitly relates to $\lambda(t,x)$ and $\nu^L$.
With this result at hand pricing of typical derivatives on the forward rate can be done in a standard way.
The input parameters for the modeler are  $\beta_{ki}$, $\gamma_{ki}$ as well as $\Sigma$, $\nu$ and $\nu^L$ while the $\alpha_{ki}$
are determined as above to ensure that the model is free of arbitrage.

\begin{proof} The idea is to identify $b$ (see \eqref{feq}) from the dynamics of $F$ and then
show that the drift condition \eqref{dc1} holds with the chosen $\alpha$: a comparison of \eqref{eq:marketF} with Corollary \ref{lem:F} and \eqref{def:ttM} yields that
\begin{align*}
\beta_{ki}(t) &= \int_{T_k}^{T_{k+1}} b(t,u,x_i) du \quad \text{ and } \quad
\gamma_{ki}(t,L_{t-};y) =\int_{T_k}^{T_{k+1}} c(t,u,x_i;y) du  .
\end{align*}
The drift in \eqref{MM:F} yields that
\begin{align} \label{eq:alphaii}
\alpha_{ki}(t) &= -\lambda(t,x_i) + \langle  \Sigma \beta_{ki}(t),b^*(s,T_{k+1},x_i) \rangle \nonumber \\[2mm]
 &+ \int_{\R^d} \left( e^{ \langle\beta_{ki}(t), z \rangle  } +e^{-\langle b^*(t,T_k,x_i),z\rangle}\big( e^{-\langle \beta_{ki}(t),z \rangle}-1\big)-1 \right) \nu(dz) \nonumber\\
 & +\int_{\cI} \left( e^{ \gamma_{ki}(y,L_{t-})  } +e^{-c^*(t,T_k,x_i;y)}\big( e^{-\gamma_{ki}(y,L_{t-})}-1\big)  \right) \ind{L_t+y \le x_i} \nu^L(dy).
\end{align}
We have that
$\beta_{ki}(t) = b^*(t,T_{k+1},x_i)- b^*(t,T_{k},x_i)$ and we use the freedom to specify $b(t,T,x_i)\equiv 0$ whenever $t>T$. This in turn gives that for $T_j < t < T_{j+1}$
$$ \beta_{ji}(t) = \int_{T_j}^{T_{j+1}} b(t,u,x_i) du =  \int_t^{T_{j+1}}b(t,u,x_i) du = b^*(t,T_{j+1},x_i) $$
such that
\begin{align}
\label{eq:bstern}
 b^*(t,T_{k+1},x_i) =  \sum_{j=\eta(t)}^k \beta_{ji}(t). \end{align}
Analogously,
\begin{align}\label{eq:cstern}
 c^*(t,T_{k+1},x_i;y) =  \sum_{j=\eta(t)}^k \gamma_{ji}(t,L_{t-};y). \end{align}
Inserting this in \eqref{eq:alphaii} gives the claim.
\end{proof}

\subsection{$(T_k,x_i)$-rate modelling}
In this section we study the case where  instead of the forward price process the $(T_k,x_i)$-rate has a simple structure.
As in the previous setting we consider the lognormal case including jumps. For simplicity we assume $\nu=0$, i.e.\ the dynamics is only driven
by a Brownian motion and $L$. In turn we will obtain the result from \citeN{BGM} as a special case. Finally, we also give the results
for market models in a single-name credit risky case.

\begin{description}
\item[(A7)]
Assume $\nu(dx)=0$ and that for each $(T_k,x_i) \in \cS$ the  functions $\beta_{ki}:\R^+\to \R^d$,
$\gamma_{ki}:\R^+\times\R^+ \to \R$ are measurable and bounded.
\end{description}

\begin{prop} Assume (A1)-(A4) and (A7) hold.  $(T_k,x_i)$-rates  given  by
\begin{align}\label{eq:marketL}
{dL(t,T_k,x_i)} & = {L(t-,T_k,x_i)} \Big( \alpha_{ki}(t) dt + \langle \beta_{ki},   dW(t) \rangle\Big) \nonumber\\
&+ \frac{1+\delta_k L(t-,T_k,x_i)} {\delta_k } \int_{\cI}  \left(e^{ \gamma_{ki}(y,L_{t-})} -1  \right)\ind{L_{t-}+y \le x_i} \mu^L(ds,dy)
\end{align}
for all $(T_k,x_i) \in \cS$ and zero otherwise constitute  an arbitrage-free market if
\begin{align}\label{eq:dcLMM}
\alpha_{ki} (t) & =   - \lambda(t,x_i)  + \sum_{j=\eta(t)}^{k} \frac{1+\delta_j L(t-,T_j,x_i)}{\delta_j L(t-,T_j,x_i)}  \langle \beta_{ji} (t), \Sigma \beta_{ki} (t) \rangle \nonumber\\
&+\frac{1+\delta_k L(t-,T_k,x_i)}{\delta_k L(t-,T_k,x_i)}  \nonumber\\& \qquad\cdot \int_{\cI} \bigg(e^{\gamma_{ki}(t;y)}+\Big(e^{-\gamma_{ki}(t,L_{t-};y)}-1\Big) \prod_{j=\eta(t)}^{k-1}e^{ -\gamma_{ji}(t,L_{t-};y)}  \bigg) \ind{L_t+y \le x_i} \nu^L(dy)
\end{align}
for all $(T_k,x_i)\in \cS$.
\end{prop}
With this result at hand pricing of typical derivatives on the $(T_k,x_i)$-rates can be done in a standard way,
see for example Section 6 in \citeN{BGM}.
As previously, the input parameters for the modeler are $\lambda$, $\beta_{ki}$ and $\gamma_{ki}$ while the $\alpha_{ki}$
are determined as above to ensure that the model is free of arbitrage.

\begin{proof}
We proceed similar to the proof of Proposition \ref{thm:marketF}.
A comparison of \eqref{eq:marketL} with  \eqref{MM:L2} yields  that
\begin{align}\label{bstern2}
\beta_{ki}(t) &=  \frac{1+\delta_k L(t-,T_k,x_i)} {\delta_k L(t-,T_k,x_i)}
\big(b^*(t,T_{k+1},x_i)-b^*(t,T_k,x_i) \big)
\end{align}
such that
\begin{align*}
  b^*(t,T_{k+1},x_i) &= b^*(t,T_k,x_i) + \beta_{ki}(t) \frac{\delta_k L(t-,T_k,x_i)}{1+\delta_k L(t-,T_k,x_i)} \\
  &= \sum_{j=\eta(t)}^k \beta_{ji}(t) \frac{\delta_j L(t-,T_j,x_i)}{1+\delta_j L(t-,T_j,x_i)},
\end{align*}
compare \eqref{eq:bstern}.
Furthermore, we have that  $\gamma_{ki}(y,L_{t-}) =C(t,T_k,T_{k+1}x_i;y)$
and hence \eqref{eq:cstern}. From the drift term in \eqref{MM:L}
we obtain that
\begin{align*} &\alpha_{ki}(t) =  -\lambda(t,x_i) +
\frac{1+\delta_k L(t-,T_k,x_i)} {\delta_k L(t-,T_k,x_i)} D(t,T_k,T_{k+1},x_i).
\end{align*}
With \eqref{bstern2} and \eqref{eq:cstern}
\begin{align*}
 D(t,&T_k,T_{k+1},x_i) = \frac {\delta_k L(t-,T_k,x_i)}  {1+\delta_k L(t-,T_k,x_i)}  \langle  \Sigma \beta_{ki}(t),b^*(t,T_{k+1},x_i) \rangle \\[2mm]
 &+\int_{\cI} \bigg(e^{\gamma_{ki}(t,L_{t-};y)} +\Big(e^{-\gamma_{ki}(t,L_{t-};y)}-1\Big) \prod_{j=\eta(t)}^{k-1}e^{ -\gamma_{ji}(t,L_{t-};y)}  \bigg) \ind{L_t+y \le x_i} \nu^L(dy)
\end{align*}
and we conclude as in Proposition \ref{thm:marketF}.
\end{proof}

\begin{rem} With $x_m=1$ and $c(\cdot,T,1;y)=1$ (from Assumption (A2)) gives the risk-free Libor market model
as a special case and \eqref{eq:dcLMM} equals Equation (2.6) in \citeN{BGM}. Also a doubly stochastic model for single-name
credit risk as in \citeN{Huehne} is a special case: consider the doubly stochastic random time $\tau$ which has an intensity $(\lambda_t)_{t \ge 0}$.
Choose $L_t:= \half \ind{\tau \le t}$. Then, with $x_1=\half$, one has that the rates of  defaultable bond satisfying,
on $\{\tau > t\}$,
\begin{align}
\frac{d\bar L(t,T_k)}{\bar L(t-,T_k)} & = \bar \alpha_{k}(t) dt + \langle \bar \beta_{k}(t),   dW(t) \rangle
\end{align}
 and zero otherwise constitute  an arbitrage-free market if
\begin{align}
\bar \alpha_{k} (t) &=
 - \lambda(t) + \Big\langle \bar \beta_{k}(t), \sum_{j=\eta(t)}^k \frac{\delta_j \bar L(t-,T_j)}{1+\delta_j \bar L(t-,T_j)} \bar\beta_{j}(t)
\Big\rangle.
 \end{align}
\end{rem}

\bibliographystyle{chicago}

\end{document}